\documentclass[prd,twocolumn,superscriptaddress,floatfix,nofootinbib]{revtex4-2}
\pdfoutput=1 
\usepackage[T1]{fontenc}
\usepackage{amsmath}
\usepackage{amssymb}
\usepackage{amsthm}
\usepackage{amsfonts}
\usepackage{graphicx}
\usepackage{enumitem}
\usepackage{bm}
\usepackage{rotating}
\usepackage{hyperref}
\usepackage{pbox}
\usepackage[dvipsnames]{xcolor}
\usepackage{array}
\usepackage{physics}
\usepackage{dsfont}
\usepackage{mathtools}
\usepackage{leftidx}
\usepackage{lmodern}
\usepackage[normalem]{ulem}
\usepackage{braket}
\usepackage{tensor}
\usepackage{mathrsfs}
\usepackage{float}
\usepackage{dsfont}
\usepackage[margin=2cm]{geometry}
\usepackage[title]{appendix}
\usepackage{tcolorbox}

\usepackage{tikz}
\usetikzlibrary{arrows.meta,decorations.pathmorphing,math}

\renewcommand{\Re}{\mathrm{Re}}
\renewcommand{\Im}{\mathrm{Im}}
\renewcommand{\bar}{\overline}
\renewcommand{\tilde}{\widetilde}
\newcommand{\skri}{\mathscr{I}}
\newcommand{\ii}{\mathsf{i}}

\newcommand{\sx}{\mathsf{x}}
\newcommand{\sy}{\mathsf{y}}

\newcommand{\rr}[1]{\left(#1\right)}
\newcommand{\bx}{{\bm{x}}}

\newcommand{\bk}{{\bm{k}}}

\newcommand{\supp}{\text{supp}}

\newcommand{\rao}{\hat{\rho}_{\textsc{a}}^{0}}
\newcommand{\rbo}{\hat{\rho}_{\textsc{b}}^{0}}

\newcommand{\rof}{\hat{\rho}_{\phi}^{0}}

\newcommand{\roo}{\hat{\rho}^{0}}

\newcommand{\R}{\mathbb{R}}
\newcommand{\C}{\mathbb{C}}
\newcommand{\M}{\mathcal{M}}

\newcommand{\A}{\mathcal{A}}
\newcommand{\W}{\mathcal{W}}
\newcommand{\CS}{C^\infty_0(\M)}
\newcommand{\Sol}{\mathsf{Sol}}

\newtheorem{lemma}{Lemma}

\begin{document}

\title{Fermi two-atom problem: non-perturbative approach via relativistic quantum information and algebraic quantum field theory}

\author{Erickson Tjoa}
\email{e2tjoa@uwaterloo.ca}
\affiliation{Department of Physics and Astronomy, University of Waterloo, Waterloo, Ontario, N2L 3G1, Canada}
\affiliation{Institute for Quantum Computing, University of Waterloo, Waterloo, Ontario, N2L 3G1, Canada}

\begin{abstract}
    In this work we revisit the famous Fermi two-atom problem, which concerns how relativistic causality impacts atomic transition probabilities, using the tools from relativistic quantum information (RQI) and algebraic quantum field theory (AQFT). The problem has sparked different analyses from many directions and angles since the proposed solution by Buchholz and Yngvason (1994). Some of these analyses employ various approximations, heuristics, perturbative methods, which tends to render some of the otherwise useful insights somewhat obscured. It is also noted that they are all studied in flat spacetime. We show that current tools in relativistic quantum information, combined with algebraic approach to quantum field theory, are now powerful enough to provide fuller and cleaner analysis of the Fermi two-atom problem for arbitrary curved spacetimes in a completely non-perturbative manner. Our result gives the original solution of Buchholz and Yngvason a very operational reinterpretation in terms of qubits interacting with a quantum field, and allows for various natural generalizations and inclusion of detector-based local measurement for the quantum field (\href{https://journals.aps.org/prd/abstract/10.1103/PhysRevD.105.065003}{PRD \textbf{105}, 065003}).

\end{abstract}

\maketitle

\section{Introduction}

Fermi's two atom problem is the thought experiment proposed by Enrico Fermi \cite{Fermi1932radiation} and attracted much attention after being revisited sixty years later by Hegerfeldt \cite{Hegerfeldt1994fermi}. The problem states that there is an apparent tension between quantum electrodynamics and Einstein causality in the following thought experiment. Suppose we have two atoms $A$ and $B$ are separated by some distance $r$. If at some time $t=0$ atom $A$ is excited, atom $B$ is in ground state, and no photon is present (the electromagnetic field is in its vacuum state). The claim is that immediately at $t=0^+$, there is nonzero probability that atom $B$ gets excited, in apparent violation of Einstein causality that imposes maximum speed of causal propagation at the speed of light. 

It was already apparent since the work of Buchholz and Yngvason \cite{Buchholz1994fermiatom} that there were some problems with Fermi's original thought experiment. In essence, the issue is that one needs to calculate the excitation probability of $B$ irrespective of the initial states of $A$ and the electromagnetic field, and furthermore without ``damaging'' approximations such as the rotating wave approximations. Furthermore, one should not invoke any procedure that involves ``projective measurement'' of the electromagnetic field state, which already leads to causality violation for related but different reason \cite{sorkin1956impossible}. Note that the proposed ``solutions'' by Hegerfeldt are perhaps asking too much to any practising experimentalists: to name one, they suggest that it may be unphysical to demand statistical independence of \textit{state preparation} of the two atoms in full quantum electrodynamics.  Overall, the consensus since the result of \cite{Buchholz1994fermiatom} is that there is no causality violation, so for all practical purposes we can say that the problem has been ``solved''.

Nonetheless, we believe that the existing analyses and solutions are somewhat unsatisfactory for at least two reasons. First, while Hegerfeldt's arguments in \cite{Hegerfeldt1994fermi} are powerful in that they only appeal to positivity of the Hamiltonian, in practical situations many state preparation procedures using realistic systems should not require one to carefully track ``soft photon cloud'' (nowadays known as preparation of dressed states) to make useful experiments\footnote{Of course there are experiments where dressing is important or if the objective is to work with dressed states (see, e.g., \cite{kato2019observation} and references therein).}  (see, e.g., \cite{Leon2011fermi}). Similarly, Buchholz and Yngvason's solutions, obtained using ``primitive causality'' and ``causal (statistical) independence'', rests on the use of expectation values of observables which would only hold in the limit of infinitely many repeated experiments \cite{hegerfeldt1994problems}. Both approaches seem to at least agree that strict localization of states are untenable. We believe that while the analyses are valid, they are perhaps too strong for illustrating relativistic causality in light-matter interaction type of scenarios, both theoretically and experimentally.

Second, many of the analyses that involve the use of non-relativistic \textit{detectors} use perturbative expansions and approximations that need to be carefully controlled. For example, Buscemi and Compagno \cite{Buscemi2009nonlocal} showed using perturbation theory that the Fermi problem analysed via a pointlike Unruh-DeWitt detector \cite{Unruh1979evaporation,DeWitt1979} is consistent with causality\footnote{Strictly speaking the calculation of probability in \cite{Buscemi2009nonlocal} is ultraviolet-divergent due to pointlike limit and sharp switching of the detector, thus some form of ``renormalization'' is required.}, while a Glauber detector \cite{Glauber1963photon}---essentially UDW detector with \textit{rotating wave approximation} (RWA)---is not. Some other related analyses that pertain Fermi problem as common theme are given in, for instance, \cite{Dickinson2016fermi, Buscemi2006localization,Borrelli2012fermi,Zohar2011fermidiscrete,Power1997energy,Menezes2017disordered}. 

In any case, subsequent studies via Unruh-DeWitt detector model, which treats two-level quantum system (``qubits'') as a detector interacting with a quantum field, more or less settle the debate in perturbative regime in flat spacetimes, with or without cavity \cite{Cliche2010channel, Jonsson2014signalQED,Leon2011fermi,Causality2015Eduardo}. These more modern and careful works are notable for being closer to quantum information theory where cleaner tools are available and experiments are quite advanced. Furthermore, it is now well-understood that within perturbative regime, there are some (mild) causality violation for detector models with \textit{finite size} that is completely separate from the Fermi problem itself, but rather related to non-relativistic nature of any model of quantum-mechanical detector \cite{Causality2015Eduardo,Charis2020causality,Maria2021faster,Tales2020GRQO,Bruno2021broken}. Modern analysis of the impact of rotating wave approximations on causality outside the context of Fermi problem has also been investigated in \cite{Funai2019RWA,Dolce2006cavity,clerk1998nonlocality}. 

In this work we would like to clean up all the existing analyses about the Fermi two-atom problem by providing a fully non-perturbative calculation using Unruh-DeWitt detector model very commonly used in \textit{relativistic quantum information} (RQI) combined with the framework of \textit{algebraic quantum field theory} (AQFT). This allows us to (1) generalize the Fermi two-atom problem to arbitrary globally hyperbolic curved spacetimes where relativistic causality still needs to hold; (2) remove any issue having to do with perturbative expansions and potentially contentious and harder-to-control approximations; and (3) give a much more economical expressions using the language in quantum information (QI) and AQFT. In fact, the economy of expression is essentially due to the fact that the calculation is very closely related to how we study quantum communication channels in quantum information theory. The approach is based on the  ``delta-coupling'' Unruh-DeWitt detector model where the detectors interact with the field very rapidly, effectively at a single instant in time (for some applications of delta coupling in RQI, see \cite{jonsson2018transmitting,tjoa2022channel,Simidzija2020transmit} for relativistic quantum communication, or \cite{Sahu2022sabotaging,Gallock2021nonperturbative,Simidzija2018nogo,Henderson2020temporal} for correlation harvesting).

At this point we should mention that by ``non-perturbative regime'' we mean that the interaction picture unitary evolution induced by the detector-field interaction Hamiltonian is a finite sum of bounded operators, without truncation of Dyson series expansion at some finite powers of coupling strengths, as is usually done in many situations where the interaction involves weak coupling. This relies on physical arguments that we can remove the time-ordering operation in this ``rapid coupling'' regime. We do \textit{not} mean that this calculation is non-perturbative in the sense that we solved \textit{exactly} for the full interacting theory of the detector-field system as a dynamical system, which is a generally difficult problem related to Haag's theorem and existence of interaction picture representations in general (see, e.g., \cite{fewster2020algebraic} and refs. therein for discussions). We will thus refer to the latter as having \textit{exact solution} for the dynamical system, while here we give a \textit{non-perturbative solution}, i.e., non-perturbative in the usual sense of having no truncation of any series expansion. In this context, we may also regard delta-coupling regime as a limit where effectively we are performing ``resummation'' of the series expansion of the unitary evolution.

As a side bonus, we will show that there is no difficulty in extending the calculations to include non-vacuum states of the field, and for quasifree states (field states with vanishing one-point functions) essentially no extra calculation needs to be done at all. While non-vacuum states may obscure the interpretation of the Fermi problem, we will see that one can really separate the causal propagation of information from the non-local correlations due to highly entangled nature of quantum field states --- this is essentially the case because causal nature of QFT is \textit{state-independent} and only the non-local correlations in the field is state-dependent \cite{Tjoa2021notharvesting}. It is part of the author's intention to bridge and popularize the two approaches from RQI and AQFT and make them work more seamlessly (see also, e.g., \cite{tjoa2022channel,Landulfo2016magnus1}). We also show briefly that local measurement theory for quantum field theory outlined in \cite{josepolo2022measurement} naturally fits the Fermi problem in the non-perturbative approach.

We remark that the approach we adopt here should give a  more complete closure to the Fermi problem in a manner that is as transparent as  possible. The only real (significant) simplification we make in this work is the use of \textit{scalar field} rather than actual electromagnetic field, hence monopole detector rather than atomic dipole (more general UDW model involving dipole-electromagnetic interaction was considered, e.g., in \cite{pozas2016entanglement,Lopp2021deloc}), since in the Fermi problem the ``directionality'' of the vector field is irrelevant to the discussion of relativistic causality. Extensions to more complicated coupling should be straightforward.

Our paper is organized as follows. In Section~\ref{sec: AQFT} we review algebraic approach to scalar QFT in $(3+1)$-dimensional globally hyperbolic Lorentzian spacetime. In Section~\ref{sec: UDW} we introduce the delta-coupled UDW detector model to set up the Fermi two-atom problem. In Section~\ref{sec: Fermi} we revisit the Fermi problem and show that it is consistent with relativistic causality in fully non-perturbative manner, valid for any globally hyperbolic curved spacetime and for any choice of quasifree Hadamard states for the field. In Section~\ref{sec: measurement} we include brief discussion on what happens to the Fermi problem setup if some actual measurement process is performed by one of the detectors. In Section~\ref{sec: conclusion} we will comment on several directions for further investigations. 

We adopt the natural units $c=\hbar=1$ and mostly-plus signature for the metric. We adopt physicists' convention of denoting Hermitian conjugate by $\dagger$, but we retain the term $C^*$-algebra where $*$-involution is mathematicians' version of Hermitian conjugation; context should make things clear.

\section{AQFT for scalar field in curved spacetime}
\label{sec: AQFT}

In order to ensure that this work is self-contained, in this section we briefly review the algebraic framework for quantization of real scalar field in arbitrary (globally hyperbolic) curved spacetimes. We will follow the convention in \cite{tjoa2022modest}. Note that in the AQFT literature there are various different conventions being used (see, e.g., \cite{wald1994quantum,Dappiaggi2005rigorous-holo,Moretti2005BMS-invar,Khavkhine2015AQFT,fewster2019algebraic}). A highly recommended accessible introduction to $*$-algebras and $C^*$-algebras for applications in this setting can be found in \cite{fewster2019algebraic} (see also \cite{hollands2017entanglement}). 

Readers are recommended to skip to Section~\ref{sec: UDW} or Section~\ref{sec: Fermi} if they are more interested in the qubit-side of the Fermi two-atom problem and less on the field-theoretic side, referring to this section when certain details need to be consulted.

\subsection{Algebra of observables and states}

We consider a free, real scalar field $\phi$ in (3+1)-dimensional globally hyperbolic Lorentzian spacetime $(\mathcal{M},g_{ab})$. The field obeys the Klein-Gordon equation
\begin{align}
     P\phi = 0\,,\quad  P = \nabla_a\nabla^a - m^2  - \xi R\,,
     \label{eq: KGE}
\end{align}
where $\xi \geq 0$, $R$ is the Ricci scalar and  $\nabla$ is the Levi-Civita connection with respect to $g_{ab}$. Global hyperbolicity ensures that $\M$ admits foliation by spacelike Cauchy surfaces $\Sigma_t$ labelled by time parameter $t$. It also guarantees that the wave equation \eqref{eq: KGE} has well-posed initial value problem throughout $\M$.

Let $f\in \CS$ be a smooth compactly supported test function on $\M$. The retarded and advanced propagators $E^\pm\equiv E^\pm(\sx,\sy)$ associated to the Klein-Gordon operator $P$ are Green's functions obeying
\begin{align}
    E^\pm f\equiv (E^\pm f)(\sx) \coloneqq \int \dd V'\, E^\pm (\sx,\sx')f(\sx') \,,
\end{align}
where we have the inhomogeneous equation $P(E^\pm f) = f$. Here $\dd V' = \dd^4\sx'\sqrt{-g}$ is the invariant volume element. The \textit{causal propagator} is defined to be the advanced-minus-retarded propagator $E=E^--E^+$. If $O$ is an open neighbourhood of some Cauchy surface $\Sigma$ and $\varphi$ is any real solution with compact Cauchy data to Eq.~\eqref{eq: KGE}, denoted by $\varphi \in \Sol_\R(\M)$, then there exists $f\in \CS$ with $\supp(f)\subset O$ such that $\varphi=Ef$ \cite{Khavkhine2015AQFT} (see also \cite{dimock1980algebras} for more details on the causal propagator).

Next, we review algebraic approach for the real scalar quantum field theory {(for comparison with canonical quantization formulation see, e.g., \cite{fewster2019algebraic,Khavkhine2015AQFT,KayWald1991theorems}).} In AQFT, the quantization of real scalar field is regarded as an $\R$-linear mapping from the space of smooth compactly supported test functions to a unital $*$-algebra $\A(\M)$ given by
\begin{align}
    \hat\phi: C^\infty_0(\mathcal{M})&\to \A(\M)\,,\quad f\mapsto \hat\phi(f)\,,
\end{align}
which obeys the following conditions:
\begin{enumerate}[leftmargin=*,label=(\alph*)]
    \item (\textit{Hermiticity}) $\hat\phi(f)^\dag = \hat\phi(f)$ for all $f\in \CS$;
    \item (\textit{Klein-Gordon}) $\hat\phi(Pf) = 0$ for all $f\in \CS$;
    \item (\textit{Canonical commutation relations}  (CCR)) $[\hat\phi(f),\hat\phi(g)] = \ii E(f,g)\openone $ for all $f,g\in \CS$, where $E(f,g)$ is the smeared causal propagator
    \begin{align}
        E(f,g)\coloneqq \int \dd V f(\sx) (Eg)(\sx)\,.
    \end{align}
    \item (\textit{Time slice axiom}) Let $\Sigma\subset \M$ be a Cauchy surface and $O$ a fixed open neighbourhood of $\Sigma$. $\A(\M)$ is generated by the unit element $\openone$ (hence $\A(\M)$ is unital) and the smeared field operators $\hat\phi(f)$ for all $f\in \CS$ with $\supp(f)\subset O$.
\end{enumerate}
The $*$-algebra $\A(\M)$ is called the \textit{algebra of observables} of the scalar field. The \textit{smeared} field operator reads
\begin{align}\label{eq: ordinary smearing}
    \hat\phi(f) = \int \dd V\hat\phi(\sx)f(\sx)\,,
\end{align}
where the object $\hat\phi(\sx)$ is to be regarded as an operator-valued distribution. 

The algebra of observables $\A(\M)$ will appear more concrete if we make its connection with the symplectic structure of the theory more explicit. First, the vector space $\Sol_\R(\M)$ can be made into a \textit{symplectic} vector space by equipping it with a symplectic form $\sigma:\Sol_\R(\M)\times\Sol_\R(\M)\to \R$, defined as
\begin{align}
    \sigma(\phi_1,\phi_2) \coloneqq \int_{\Sigma_t}\!\! {\dd\Sigma^a}\,\Bigr[\phi_{{1}}\nabla_a\phi_{{2}} - \phi_{{2}}\nabla_a\phi_{{1}}\Bigr]\,,
    \label{eq: symplectic form}
\end{align}
where $\dd \Sigma^a = -t^a \dd\Sigma$, $-t^a$ is the inward-directed unit normal to the Cauchy surface $\Sigma_t$, and $\dd\Sigma = \sqrt{h}\,\dd^3\bx$ is the induced volume form on $\Sigma_t$ \cite{Poisson:2009pwt,wald2010general}. This definition is independent of the Cauchy surface. With this, we can regard $\hat\phi(f)$ as \textit{symplectically smeared field operator}  \cite{wald1994quantum} 
\begin{align}
    \label{eq: symplectic smearing}
    {\hat\phi(f) \equiv \sigma(Ef,\hat\phi)\,,}
\end{align}
and the CCR algebra can be written as 
\begin{align}
    {[\sigma(Ef,\hat\phi),\sigma(Eg,\hat\phi)] = \ii\sigma(Ef,Eg)\openone = \ii E(f,g)\openone \,,}
\end{align}
where $\sigma(Ef,Eg) = E(f,g)$ in the second equality follows from Eq.~\eqref{eq: ordinary smearing} and \eqref{eq: symplectic smearing}. The symplectic smearing has the advantage of keeping the dynamical content manifest at the level of the field operators (via the causal propagator). Thinking of the field operator as symplectically smeared field operator can be useful in some contexts, such as when studying scalar QFT at $\skri^+$ (see, e.g., \cite{Landulfo2021cost,tjoa2022modest}).

It is often more convenient to work with the ``exponentiated'' version of $\A(\M)$ called the \textit{Weyl algebra} (denoted by $\W(\M)$), since its elements are (formally) bounded operators. The Weyl algebra $\W(\M)$ is a unital $C^*$-algebra generated by elements which formally take the form 
\begin{align}
    W(Ef) \equiv 
    {e^{\ii\hat\phi(f)}}\,,\quad f\in \CS\,.
    \label{eq: Weyl-generator}
\end{align}
These elements satisfy \textit{Weyl relations}:
\begin{equation}
    \begin{aligned}
    W(Ef)^\dagger &= W(-Ef)\,,\\
    W(E (Pf) ) &= \openone\,,\\
    W(Ef)W(Eg) &= e^{-\frac{\ii}{2}E(f,g)} W(E(f+g))
    \end{aligned}
    \label{eq: Weyl-relations}
\end{equation}
where $f,g\in \CS$. Note that the \textit{microcausality} condition (also known as \textit{relativistic causality} or \textit{Einstein causality}) can be seen as follows: using the Weyl relations Eq.~\eqref{eq: Weyl-relations} and the fact that $\supp(Ef)\subset J^+(\supp(f))$ where $J^+(\supp(f))$ is the causal future of $\supp(f)$, we have \cite{Dappiaggi2005rigorous-holo}
\begin{align}
    [W(Ef),W(Eg)] = 0 \,,
    \label{eq: Weyl-commutator-bulk}
\end{align}
whenever $\supp(f)\,\cap\,\supp(g) = \emptyset$ (supports of $f$ and $g$ are causally disjoint, i.e., ``spacelike separated'')\footnote{Abstractly, we would consider elements of the Weyl algebra to be $W(\phi)$ for some $\phi\in \Sol_\R(\M)$. In this form, however, microcausality is less obvious because the third Weyl relation would have read $W(\phi_1)W(\phi_2) = e^{{-}\ii\sigma(\phi_1,\phi_2)/2}W(\phi_1 + \phi_2)$. }.

After specifying the algebra of observables, we need to provide a quantum state for the field theory. In AQFT the state is called an \textit{algebraic state}, defined by a $\C$-linear functional $\omega:\A(\M)\to \C$ (similarly for $\W(\M)$) such that 
\begin{align}
    \omega(\openone) = 1\,,\quad  \omega(A^\dagger A)\geq 0\quad \forall A\in \A(\M)\,.
    \label{eq: algebraic-state}
\end{align}
That is, a quantum state is normalized to unity and positive-semidefinite operators have non-negative expectation values. The state $\omega$ is pure if it cannot be written as $\omega= \alpha \omega_1 + (1-\alpha)\omega_2$ for any $\alpha\in (0,1)$ and any two algebraic states $\omega_1,\omega_2$; otherwise we say that the state is mixed. For the Weyl algebra, the same definition works but with element $A\in \W(\M)$.

The connection to the usual notion of Hilbert space (Fock space) in canonical quantization comes from the Gelfand-Naimark-Segal (GNS) reconstruction theorem \cite{wald1994quantum,Khavkhine2015AQFT,fewster2019algebraic}. This says that we can construct a \textit{GNS triple}\footnote{Strictly speaking we also need to provide a dense subset $\mathcal{D}_\omega\subset \mathcal{H}_\omega$ since the field operators are unbounded operators.} $(\mathcal{H}_\omega, \pi_\omega,{\ket{\Omega_\omega}})$, where $\pi_\omega: \mathcal{\A(\M)}\to {\text{End}(\mathcal{H}_\omega)}$ is a Hilbert space representation with respect to state $\omega$ such that any algebraic state $\omega$ can be realized as a \textit{vector state} {$\ket{\Omega_\omega}\in\mathcal{H}_\omega$}. The observables $A\in \A(\M)$ are then represented as operators $\hat A\coloneqq \pi_\omega(A)$ acting on the Hilbert space. With the GNS representation, the action of algebraic states take the familiar form
\begin{align}
    \omega(A) = \braket{\Omega_\omega|\hat A|\Omega_\omega}\,.
\end{align}
The main advantage of the AQFT approach is that it is independent of the representations of the CCR algebra chosen: there are as many representations as there are algebraic states $\omega$. Since QFT in curved spacetimes admits infinitely many unitarily inequivalent representations of the CCR algebra, this allows us to work with all of them simultaneously until the very last step.

For the Weyl algebra, the algebraic state and GNS representation gives concrete realization of ``exponentiation of $\hat\phi(f)$'' as a bounded operator acting on the Hilbert space. We remind the reader that the exponentiation in Eq.~\eqref{eq: Weyl-generator} is only formal as we cannot literally regard the smeared field operator $\hat\phi(f)$ as the derivative $\partial_t\bigr|_{t=0}W(t Ef)$. This is because the Weyl algebra itself does not have the right topology for this to work out \cite{fewster2019algebraic}; however, one can take the derivative at the level of the GNS representation: that is, if $\Pi_\omega:\W(\M)\to \mathcal{B}(\mathcal{H}_\omega)$ is a GNS representation with respect to $\omega$, then we do have 
\begin{align}
    \pi_\omega(\hat\phi(f)) &= -\ii\frac{\dd}{\dd t}\Bigg|_{t=0}\!\!\!\!\!\!\!\Pi_\omega(e^{\ii t \hat\phi(f)}) \equiv - \ii\frac{\dd}{\dd t}\Bigg|_{t=0}\!\!\!\!\!\!\!e^{\ii t \pi_\omega(\hat\phi(f))} \,,
\end{align}
where now $\hat\phi(f)$ is smeared field operator acting on Hilbert space $\mathcal{H}_\omega$. This is then taken to \textit{define} the formal $n$-point functions to be the expectation value in its GNS representation. For example, \textit{the smeared Wightman two-point function} reads
\begin{align}\label{eq: Wightman-formal-bulk}
    &\omega \bigr(\hat\phi(f)\hat\phi(g)\bigr) \coloneqq \braket{\Omega_\omega|\pi_\omega(\hat\phi(f))\pi_\omega(\hat\phi(g))|\Omega_\omega}\notag \\
    &\equiv -\frac{\partial^2}{\partial s\partial t}\Bigg|_{s,t=0}\!\!\!\!\!\!\!\!\braket{\Omega_\omega|e^{\ii s\pi_\omega(\hat\phi(f))}e^{\ii t\pi_\omega(\hat\phi(g))}|\Omega_\omega}\,,
\end{align}
which we also denote by $\mathsf{W}(f,g)$. In what follows we will write the formal two-point functions $\omega(\hat\phi(f)\hat\phi(g))$ with this understanding that the actual calculation is (implicitly) done with respect to the GNS representation in question: for instance, in Minkowski spacetime the vacuum GNS representation would give the Minkowski vacuum $\ket{\Omega_\omega} = \ket{0_\textsc{M}}$.

It is worth noting that not every field state is expressible as a density matrix in a particular GNS representation. For example, thermal state at temperature $\beta^{-1}$ cannot be written as $\hat\rho_\beta = e^{-\beta\hat H}/Z$ since the partition function $Z$ is divergent and hence the expression is only formal. In the algebraic framework, this problem is avoided since the Kubo-Martin-Schwinger (KMS) condition \cite{Kubo1957thermality,Martin-Schwinger1959thermality} characterizes thermal states via the correlation functions. All we need to know is how to write the corresponding algebraic state $\omega_\beta(\cdot)$. In general, given a GNS representation $\pi_\omega$ associated to an algebraic state $\omega$, the family of density matrices $\hat\rho\in \mathscr{D}(\mathcal{H}_\omega)$ in that GNS representation defines a family of algebraic states $\hat\rho \mapsto \omega_{\hat{\rho}}$: such family of states $\{\omega_{\hat\rho}\}$ for which a valid density matrix exists in $\mathcal{H}_\omega$ are called \textit{normal states}\footnote{Not all algebraic states associated to a GNS representation are expressible in terms of a density matrix: thermal states $\omega_\beta$ are therefore not a normal state of $\Pi_\omega(\W(\M))$ (viewed as a von Neumann algebra) where $\Pi_\omega$ is the vacuum representation \cite{hollands2017entanglement}.} in the folium of $\Pi_\omega$ of the Weyl algebra \cite{hollands2017entanglement}.

\subsection{Quasifree states}

In the AQFT approach there are too many algebraic states and not all of them are physically relevant. The general consensus is that physically reasonable states $\omega$ should fall under the class of \textit{Hadamard states} \cite{Khavkhine2015AQFT,KayWald1991theorems,Radzikowski1996microlocal}. Roughly speaking, these states have the right ``singular structure'' at short distances that respects the local flatness property in general relativity and has finite expectation values over all observables (see \cite{KayWald1991theorems} and references therein for more technical details). In this work, we would like to work with Hadamard states that are also \textit{quasifree}, denoted by $\omega_\mu$: these are the states which can be completely described only their two-point correlators\footnote{This means all odd-point functions vanish, non-vanishing two-point functions $\omega(\hat\phi(f)\hat\phi(g))\neq 0$, and all even-point functions can be written as linear combination of products of two-point functions. The term \textit{Gaussian states} are sometimes reserved for those that can be completely specified by its one-point and two-point functions.}. Well-known field states such as the vacuum states and thermal states are all quasifree states, with thermal states (thermality defined according to Kubo-Martin-Schwinger (KMS) condition \cite{KayWald1991theorems}) being an example of mixed quasifree state.

Let us now review the construction of quasifree state following \cite{tjoa2022channel} (based on \cite{KayWald1991theorems,Khavkhine2015AQFT,fewster2019algebraic}). Any quasifree state $\omega_\mu$ is associated to a {real inner product} $\mu:\Sol_\R(\M)\times \Sol_\R(\M)\to \R$ obeying the inequality 
\begin{align}
    |\sigma(Ef,Eg)|^2\leq 4\mu(Ef,Ef)\mu(Eg,Eg)\,,
    \label{eq: real-inner-product}
\end{align}
for any $f,g\in \CS$. The state is pure if it saturates the above inequality appropriately \cite{wald1994quantum}. Then the quasifree state $\omega_\mu$ is defined as 
\begin{align}
    \omega_\mu(W(Ef)) \coloneqq e^{-\mu(Ef,Ef)/2}\,.
    \label{eq: quasifree}
\end{align}
We will drop the subscript $\mu$ and simply write $\omega$ in what follows (unless otherwise stated). Note that Eq. \eqref{eq: quasifree} is not very useful since it does not give us a way to evaluate $\mu(Ef,Ef)$.

In order to obtain practical expression for the norm-squared $||Ef||^2 \coloneqq \mu(Ef,Ef)$, we  make $\Sol_\R(\M)$ into a Hilbert space in the following sense\footnote{We will assume that the Hilbert space is already completed via its inner product.}. In \cite{KayWald1991theorems} it was shown that we can always construct a  \textit{one-particle structure} associated to quasifree state $\omega_\mu$, namely a pair $(K,\mathcal{H})$, where $\mathcal{H}$ is  a Hilbert space $(\mathcal{H},\braket{\cdot,\cdot}_\mathcal{H})$ together with an $\R$-linear map $K:\Sol_\R(\M)\to \mathcal{H}$ such that for $\phi_1,\phi_2\in \Sol_\R(\M)$
\begin{enumerate}[leftmargin=*,label=(\alph*)]
    \item $K\Sol_\R(\M)+\ii K\Sol_\R(\M)$ is dense in $\mathcal{H}$;
    \item $\mu(\phi_1,\phi_2)=\Re\braket{K\phi_1,K\phi_2}_\mathcal{H}$;
    \item $\sigma(\phi_1,\phi_2) = 2\Im\braket{K\phi_1,K\phi_2}_\mathcal{H}$.
\end{enumerate}
In the language of canonical quantization, the linear map $K$ projects out the ``positive frequency part'' of real solution to the Klein-Gordon equation. The smeared Wightman two-point function $\mathsf{W}(f,g)$ is then related to $\mu,\sigma$ by \cite{KayWald1991theorems,fewster2019algebraic}
\begin{align}
    \hspace{-0.05cm}\mathsf{W}(f,g) &\coloneqq \omega(\hat\phi(f)\hat\phi(g)) = \mu(Ef,Eg) + \frac{\ii}{2}E(f,g)\,,
\end{align}
where we have used the fact that $\sigma(Ef,Eg) = E(f,g)$. We also  have $E(f,f)=0$, thus we get
\begin{align}
    ||Ef||^2 =  \mathsf{W}(f,f) = \braket{KEf,KEf}_\mathcal{H}\,.
    \label{eq: algebraic-norm}
\end{align}
Therefore, we can compute $\mu(Ef,Ef)$ if either (i) we know the (unsmeared) Wightman two-point distribution of the theory associated to some quantum field state, or (ii) we know the inner product $\braket{\cdot,\cdot}_\mathcal{H}$ and how to project using $K$.

The inner product $\braket{\cdot,\cdot}_\mathcal{H}$ is precisely the \textit{Klein-Gordon inner product} $(\cdot,\cdot)_{\textsc{kg}}:\Sol_\C(\M)\times\Sol_\C(\M)\to \mathbb{C}$ restricted to $\mathcal{H}$, defined by
\begin{align}
    (\phi_1,\phi_2)_\textsc{kg} \coloneqq \ii \sigma(\phi_1^*,\phi_2)\,,
    \label{eq: KG-inner-product}
\end{align}
where the symplectic form is now extended to \textit{complexified} solution $\Sol_\C(\M)$ of the Klein-Gordon equation. The restriction to $\mathcal{H}$ is necessary since $(\cdot,\cdot)_\textsc{kg}$ is not an inner product on $\Sol_\C(\M)$. In particular, we have
\begin{align}
    \Sol_\C(\M) \cong \mathcal{H}\oplus \overline{\mathcal{H}}\,,
\end{align}
where $\overline{\mathcal{H}}$ is the complex conjugate Hilbert space of $\mathcal{H}$ \cite{wald1994quantum}. It follows that Eq.~\eqref{eq: quasifree} can be written as
\begin{align}
    \omega_\mu(W(Ef)) =   e^{-{\frac{1}{2}}\mathsf{W}(f,f)} = e^{-{\frac{1}{2}}||KEf||^2_{\textsc{kg}}}\,.
    \label{eq: norm-Ef}
\end{align}

The expression in Eq.~\eqref{eq: norm-Ef} gives us a concrete way to calculate $||Ef||^2$ more explicitly. We know that the (unsmeared) vacuum Wightman two-point distribution can be calculated in canonical quantization, which reads
\begin{align}
    \mathsf{W}(\sx,\sy) &= \int \dd^3\bk\, u^{\phantom{*}}_\bk(\sx) u^*_\bk(\sy)\,,
\end{align}
where $u_\bk(\sx)$ are (positive-frequency) modes of Klein-Gordon operator $P$ normalized with respect to Klein-Gordon inner product Eq.~\eqref{eq: KG-inner-product}:
\begin{equation}
    \begin{aligned}
    (u_\bk,u_{\bk'})_\textsc{kg} &= \delta^3(\bk-\bk')\,,\quad (u_\bk^{\phantom{*}},u^*_{\bk'})_\textsc{kg} = 0\,,\\
    (u_\bk^*,u^*_{\bk'})_\textsc{kg} &= -\delta^3(\bk-\bk')\,.
    \end{aligned}
    \label{eq: KG-normalization}
\end{equation}
If we know the set $\{u_\bk\}$, we can  calculate the symmetrically smeared two-point function
\begin{align}
    \mathsf{W}(f,f) = \int \dd V\,\dd V' f(\sx)f(\sy)\mathsf{W}(\sx,\sy)\,.
    \label{eq: Wightman-double-smeared}
\end{align}
In other words, for all practical purposes knowing the Klein-Gordon norm of $KEf$ or the symmetrically smeared Wightman two-point functions is enough to completely work out many calculations for quasifree states. One can also extend the calculation to include complex smearing (see, e.g., \cite{fewster2019algebraic}, for this).

\section{Delta-coupled Unruh-DeWitt (UDW) detector model}
\label{sec: UDW}

In this section we first review the covariant generalization of the Unruh-DeWitt (UDW) detector model, as done in the spirit of \cite{Tales2020GRQO,Bruno2021broken,tjoa2022channel}. The UDW detector model is a simplification of light-matter interaction where the dipole-electromagnetic interaction is reduced to a scalar version. This model gives a good approximation of light-matter interaction when no angular momentum is exchanged (since scalar field has no vectorial component) \cite{pozas2016entanglement}.

\subsection{Covariant UDW detector model}

Consider two observers Alice and Bob each carrying an Unruh-DeWitt (UDW) detector in spacetime $\M$. The UDW detector is a two-level system (``a qubit'') with free Hamiltonian given by
\begin{align}
    \mathfrak{h}_j &= \frac{\Omega_j}{2}(\hat\sigma^z_j+\openone)\,,\quad j=A,B
\end{align}
where $\hat{\sigma}^z_j$ is the usual Pauli-$Z$ operator for detector $j$, whose ground and excited states $\ket{g_j},\ket{e_j}$ have energy $0,\Omega_j$ respectively. Let $\tau_j$ be the proper time of detector $j$ whose centre of mass travels along the worldline $\sx_j(\tau_j)$. \textit{A priori} the proper times may not coincide (i.e., the sense that $\dd\tau_A/\dd\tau_B\neq 1$) due to relativistic redshift caused by relative motion or different gravitational potential.

The covariant generalization of the UDW model is given by the following interaction Hamiltonian four-form (in interaction picture) \cite{Tales2020GRQO}
\begin{align}
    h_{I,j}(\sx) &= \dd V\,f_j(\sx)\hat\mu_j(\tau_j(\sx))\otimes\hat\phi(\sx)\,,
\end{align}
where $\dd V = \dd^4\sx \sqrt{-g}$ is the invariant volume element in $\M$, $f_j(\sx)\in\CS$ prescribes the interaction region between detector $j$ and the field. The monopole moment of detector $j$, denoted $\hat\mu_j(\tau_j)$, is given by
\begin{align}
    \hat\mu_j(\tau_j) &= \hat \sigma^x_j(\tau_j) = \hat\sigma^+_j e^{\ii\Omega_j \tau_j} + \hat\sigma^-_je^{-\ii\Omega_j \tau_j}\,,
\end{align}
where $\hat\sigma^\pm$ are the $\mathfrak{su}(2)$ ladder operators with $\hat\sigma^+_j\ket{g_j}=\ket{e_j}$ and $\hat\sigma^-_j\ket{e_j}=\ket{g_j}$. 

The total unitary time-evolution for the detector-field system is given by the time-ordered exponential (in interaction picture) \cite{Tales2020GRQO,Bruno2021broken}
\begin{align}
    U = \mathcal{T}_t\exp\left[-\ii \int_\M \dd V h_{I,A}(\sx)+h_{I,B}(
    \sx) \right]\,.
\end{align}
Since the spacetime is globally hyperbolic, we can take the time ordering to be with respect to the global time function $t$. Without loss of generality we can set $t(\tau_{\textsc{a},0}) = \tau_{\textsc{a},0}=0$. 

At this point, depending on the problem at hand one may proceed to evaluate the time evoluation perturbatively or non-perturbatively. There is a great deal of flexibility when one chooses to work within perturbative regime, but there is mild causality violation and ``broken covariance'' whose origin can be traced to the combination of time-ordering and non-relativistic nature of the detector model \cite{Tales2020GRQO,Bruno2021broken}\footnote{As articulated in \cite{Yngvason2005vonneumann}, this should not come as a surprise or a source of concern, so long as one is aware of the regime of validity and approximations being used.}. Since our goal is to reformulate the Fermi problem in the clearest possible manner, we will adopt a non-perturbative approach based on \textit{delta-coupling}. 

\subsection{Delta coupling and the choice of coordinate systems}
The delta-coupling regime for UDW model is the regime where the interaction timescale between the detector and the field is assumed to be much faster than all the relevant timescales of the problem: that is, we model the interaction with each detector as taking place at a single instant in time (with respect to some time function, typically the detector's proper time or the global time function). This regime is particularly suited for the analysis of the Fermi two-atom problem. 

Mathematically, each detector interacts at some fixed $\tau_j=\tau_{j,0}$ in its own centre-of-mass rest frame. In terms of the detector's proper frame (e.g., using Fermi normal coordinates \cite{Tales2020GRQO,Bruno2021broken,Tales2022FNC}) the spacetime smearing can be factorized into products of switching functions and its spatial profile:
\begin{align}
    f_j(\sx) = \lambda_j\eta_j\delta(\tau_j-\tau_{j,0})F_j(\bar{\bx})
    \label{eq: delta-smearing}\,,
\end{align}
where we have written $\sx=(\tau,\bar{\bx})$ in the local coordinates of the detector and $F_j(\bar{\bx})$ gives the spatial profile of the detector. In FNC we also have that the centre of mass trajectory is $\sx_j(\tau)\equiv (\tau,\bm{0})$.

For the Fermi problem, we require that we can order Alice and Bob's local operations as being time-ordered with respect to some common time function $t$. Since the supports of Alice and Bob are assumed to not overlap due to compactly supported interaction regions, we can always do this unambiguously. In what follows we assume that the two detectors are on  static trajectories so that $t_j = t(\tau_j)$ (independent of the spatial coordinates $\bx$). In this case, it means that we need to be able to have $t(\tau_{\textsc{a},0})<t(\tau_{\textsc{b},0})$. Using the global time function $t$ coming from the foliation $\M\cong \R\times \Sigma_t$, the time function can be used to order the two local unitaries for each detector, so that the unitary $U$ reduces to a product of two simple unitaries $U=U_{\textsc{b}}U_{\textsc{a}}$, where
\begin{align}
    U_{j} &= \exp\left[-\ii \hat\mu_j(\tau_{j,0}) \otimes \hat Y_j \right]\,,\label{eq: unitary-delta}\\
    \hat Y_j &= \hat\phi(f_j) = \int_\M\dd V\,f_j(\sx)\hat\phi(\sx) \,,
    \label{eq: smeared-Y-free}
\end{align}
with $f_j(\sx)$ is given in Eq.~\eqref{eq: delta-smearing} and $t(\tau_{\textsc{a},0})<t(\tau_{\textsc{b},0})$. This unitary can be written as a sum of bounded operators
\begin{align}
    U_j = \openone \otimes \cos\hat Y_j-\ii\hat\mu_j(\tau_{j,0})\otimes\sin\hat Y_j\,,
\end{align}
where the smeared operator $\hat Y_j$ should read
\begin{align}
    \hat Y_j &= \lambda_j\eta_j\int_{\tau=\tau_{j,0}}\!\!\!\!\!\!\dd\Sigma \,F(\bar{\bx})\hat\phi(\sx(\tau_{j,0},\bar{\bx})) \,,
    \label{eq: smeared-Y}
\end{align}
where $\dd\Sigma$ is the induced volume element for constant-$\tau_j$ slices associated to the rest frame of the detectors.  We will leave it as $\hat Y_j = \hat \phi(f_j)$ to keep it manifestly coordinate-independent in the sense of Eq.~\eqref{eq: smeared-Y-free}, and indeed we do not need any detailed description of $\hat Y_j$ in what follows. For simplicity we assume that the constants $\eta_j$ (with dimension of [Length]) are equal.

Finally, we should state for completeness that the results of this work using delta-coupled detector is not strictly speaking at the level of rigour of pure mathematics. This is because the spacetime smearing $f_j(\sx)$ for delta coupling in Eq.~\eqref{eq: delta-smearing} does not belong to $\CS$. We have in mind physical intuition that delta coupling corresponds to very rapid interaction, i.e., an approximation of smooth compactly supported functions that are very localized in time\footnote{This is in the same spirit as to how Gaussian spacetime smearing, which is also not in $\CS$, gives very good approximation to compactly supported smearing functions because the tails are exponentially suppressed.}. Justifying this fully rigorously would take us too far away; however, one can adopt the gapless detector model as done in \cite{Landulfo2016magnus1} and make analogous calculations in what follows. The results can be shown to be analogous, except that now the notion of ground and excited states are harder to see because the energy levels are degenerate. Therefore, the use of AQFT in this work is mainly to ensure that the setting generalizes easily to arbitrary curved spacetimes and for any of the GNS representations, since canonical quantization somewhat obscures the essential physics of the Fermi two-atom problem.

\section{Fermi two-atom problem and relativistic causality}
\label{sec: Fermi}

We are now ready to set up the Fermi two-atom problem. Let us consider the initial state of the detector-field system to be given by initially uncorrelated state
\begin{align}
    \roo &= \rao \otimes \rbo \otimes \rof\,.
\end{align}
We are interested in what happens to Bob's detector given what happens to Alice's detector. Thus we are interested in the quantum channel $\Phi:\mathscr{D}(\mathcal{H}_\textsc{b})\to \mathscr{D}(\mathcal{H}_\textsc{b})$, where $\mathscr{D}(\mathcal{H}_\textsc{b})$ is the space of density matrices associated to Hilbert space of Bob's detector $\mathcal{H}_\textsc{b}$. The channel is naturally defined in the Stinespring representation
\begin{align}
    \Phi(\rbo) = \tr_{\textsc{a},\phi}(U\roo U^\dagger)\,. 
\end{align}
Note that this is different from the context of quantum communication where we are interested instead in the channel acting on Alice's state mapping into Bob's Hilbert space, which would be some channel $\mathcal{E}:\mathscr{D}(\mathcal{H}_\textsc{a}) \to \mathscr{D}(\mathcal{H}_\textsc{b})$.

We can calculate the channel explicitly in closed form. Since the expression for the full density matrix will be very useful for future calculations and in other contexts \cite{tjoa2022singlequbit}, we will expand them completely here. It is useful to organize the calculations in terms of $\rbo$:
\begin{widetext}
\begin{align}
    &\hat{\rho}_{\textsc{ab}\phi} = U(\rao\otimes\rbo\otimes\rof)U^\dagger \notag\\
    &\!\!= \rbo\otimes\rr{\rao\otimes C_\textsc{b}C_\textsc{a}\rof C_\textsc{a}C_\textsc{b} +
    \hat\mu_\textsc{a}\rao\hat\mu_\textsc{a}\otimes C_\textsc{b}S_\textsc{a}\rof S_\textsc{a}C_\textsc{b} +
    \ii\rao\hat\mu_\textsc{a}\otimes C_\textsc{b}C_\textsc{a}\rof S_\textsc{a}C_\textsc{b} -
    \ii\hat\mu_\textsc{a}\rao\otimes C_\textsc{b}S_\textsc{a}\rof C_\textsc{a}C_\textsc{b} }\notag\\
    &\!\!+ \hat\mu_\textsc{b}\rbo\hat\mu_\textsc{b}\otimes
    \rr{\rao\otimes S_\textsc{b}C_\textsc{a}\rof C_\textsc{a}S_\textsc{b} +
    \hat\mu_\textsc{a}\rao\hat\mu_\textsc{a}\otimes S_\textsc{b}S_\textsc{a}\rof S_\textsc{a}S_\textsc{b} +
    \ii\rao\hat\mu_\textsc{a}\otimes S_\textsc{b}C_\textsc{a}\rof S_\textsc{a}S_\textsc{b} -
    \ii\hat\mu_\textsc{a}\rao\otimes S_\textsc{b}S_\textsc{a}\rof C_\textsc{a}S_\textsc{b} }\notag\\
    &\!\!+ \ii \rbo\hat\mu_\textsc{b}\otimes
    \rr{\rao\otimes C_\textsc{b}C_\textsc{a}\rof C_\textsc{a}S_\textsc{b} +
    \hat\mu_\textsc{a}\rao\hat\mu_\textsc{a}\otimes C_\textsc{b}S_\textsc{a}\rof S_\textsc{a}S_\textsc{b} +
    \ii\rao\hat\mu_\textsc{a}\otimes C_\textsc{b}C_\textsc{a}\rof S_\textsc{a}S_\textsc{b} -
    \ii\hat\mu_\textsc{a}\rao\otimes C_\textsc{b}S_\textsc{a}\rof C_\textsc{a}S_\textsc{b} }\notag\\
    &\!\!-\ii \hat\mu_\textsc{b}\rbo\otimes
    \rr{\rao\otimes S_\textsc{b}C_\textsc{a}\rof C_\textsc{a}C_\textsc{b} +
    \hat\mu_\textsc{a}\rao\hat\mu_\textsc{a}\otimes S_\textsc{b}S_\textsc{a}\rof S_\textsc{a}C_\textsc{b} +
    \ii\rao\hat\mu_\textsc{a}\otimes S_\textsc{b}C_\textsc{a}\rof S_\textsc{a}C_\textsc{b} -
    \ii\hat\mu_\textsc{a}\rao\otimes S_\textsc{b}S_\textsc{a}\rof C_\textsc{a}C_\textsc{b} }\,,
\end{align}
\end{widetext}
using the shorthand $C_j \equiv \cos\hat Y_j$ and $S_j\equiv \sin\hat Y_j$, and here it is understood that $\hat\mu_j\equiv \hat\mu_j(\tau_{j,0})$ in order to alleviate notation. For an algebraic state associated to $\rof$ (which defines the distinguished folium of normal states associated to some algebraic state $\mathcal{H}_\omega$, \textit{c.f.} Section~\ref{sec: AQFT} or \cite{hollands2017entanglement}) let us define another shorthand\footnote{This shorthand is slightly different from the definition given in \cite{tjoa2022channel} but simpler to use in terms of the ordering of the operators.}
\begin{equation}
    \begin{aligned}
    \gamma_{ijkl} &\coloneqq \omega(X^{(i)}_\textsc{a} X^{(j)}_\textsc{b}X^{(k)}_\textsc{b}X^{(l)}_\textsc{a}) \\
    &\equiv \tr(\rof X^{(i)}_\textsc{a} X^{(j)}_\textsc{b}X^{(k)}_\textsc{b}X^{(l)}_\textsc{a}) 
    \end{aligned}
\end{equation}
where $i,j,k,l = c,s$ for cosine and and sine respectively, e.g., $X^{(c)}_\textsc{a}\equiv \cos\hat Y_\textsc{a}$. By taking partial trace over $A$ and $\phi$, we get
\begin{align}
    \Phi(\rbo) &= (\gamma_{cccc}+\gamma_{sccs}+\ii\alpha (\gamma_{sccc}-\gamma_{cccs}) )\rbo  \notag\\
    &+ (\gamma_{cssc}+\gamma_{ssss}+\ii\alpha (\gamma_{sssc}-\gamma_{csss}) ) \hat\mu_\textsc{b}\rbo\hat\mu_\textsc{b}\notag\\
    &+ \ii(\gamma_{cscc}+\gamma_{sscs}+\ii\alpha (\gamma_{sscc}-\gamma_{cscs}) ) \rbo\hat\mu_\textsc{b}\notag\\
    &- \ii(\gamma_{ccsc}+\gamma_{scss}+\ii\alpha (\gamma_{scsc}-\gamma_{ccss}) ) \hat\mu_\textsc{b}\rbo\,,
    \label{eq: channel-Bob}
\end{align}
where $\alpha = \tr(\hat\mu_\textsc{a}\rao)$. The coefficients can be computed straightforwardly for quasifree or Gaussian states using only properties of the Weyl algebra and Weyl relations \eqref{eq: Weyl-relations}. Note that as stated in Eq.~\eqref{eq: channel-Bob}, the result is valid for \textit{any} initial state of Alice and Bob, and furthermore it is also valid for any initial algebraic state of the field.

For the Fermi two-atom problem, the channel simplifies considerably because we are going to consider the field to be in the vacuum state, which is a subclass of quasifree Hadamard states. We will also set Alice's state to be in the excited state $\rao = \ketbra{e_\textsc{a}}{e_\textsc{a}}$. This gives us $\alpha=1$, and furthermore $\gamma_{ijkl}=0$ if there are odd number of sines and cosines. The channel thus reduces to
\begin{align}
    \Phi(\rbo) &= (\gamma_{cccc}+\gamma_{sccs}) \rbo + (\gamma_{cssc}+\gamma_{ssss}) \hat\mu_\textsc{b}\rbo\hat\mu_\textsc{b}\notag\\
    &-\ii(\gamma_{sscc}-\gamma_{cscs})  \rbo\hat\mu_\textsc{b} + \ii (\gamma_{scsc}-\gamma_{ccss}) \hat\mu_\textsc{b}\rbo\,.
    \label{eq: channel-Bob-quasifree}
\end{align}

Let us compute the coefficients. We first prove useful identities below, which will give very efficient shortcut for the computation.
\begin{lemma}
    We have the ``twisted'' product-to-sum formulae for Weyl algebra:
    \begin{subequations}
    \begin{align}
        2C_iC_j &= C_{i+j}e^{-\ii E_{ij}/2} + C_{i-j}e^{\ii E_{ij}/2} \,,\\
        -2S_iS_j &= C_{i+j}e^{-\ii E_{ij}/2} - C_{i-j}e^{\ii E_{ij}/2} \,,\\
        2C_iS_j &= S_{i+j}e^{-\ii E_{ij}/2} - S_{i-j}e^{\ii E_{ij}/2} \,,\\
        2S_iC_j &= S_{i+j}e^{-\ii E_{ij}/2} + S_{i-j}e^{\ii E_{ij}/2} \,,
    \end{align}
    \end{subequations}
    where $C_{i\pm j}\equiv \cos(\hat\phi(f_i\pm f_{j}))$, $S_{i\pm j}\equiv \sin(\hat\phi(f_i\pm f_{j}))$ and $E_{ij}\coloneqq E(f_i,f_j)$ is the smeared causal propagator. If $\supp(f_i)$ and $\supp(f_j)$ are spacelike separated, we have $E_{ij}=0$ and these reduce to the standard product-to-sum formula in trigonometry for complex numbers (or for commuting operators).
\end{lemma}
\begin{proof}
    We just prove for one case and the rest follows analogously. In terms of the Weyl algebra we have the algebraic exponential of the field operator
    \begin{align}
        W(Ef)\equiv e^{\ii\hat\phi(f)}\,,
    \end{align}
    so that $C_i = \frac{1}{2}(W(Ef_i)+W(-Ef_i))$. We have
    \begin{align}
        2C_iC_j &= \frac{1}{2}\Bigr[W(Ef_i)W(Ef_j)+W(-Ef_i)W(-Ef_j)\notag\\
        &\hspace{0.4cm}+W(Ef_i)W(-Ef_j)+W(-Ef_i)W(Ef_j)\Bigr]\notag\\
        &=\frac{1}{2}\Bigr[W(E(f_i+f_j))+W(-E(f_i+f_j))\Bigr]e^{-\ii E_{ij}/2}\notag\\
        &+\frac{1}{2}\Bigr[W(E(f_i-f_j))+W(-E(f_i-f_j))\Bigr]e^{\ii E_{ij}/2}\notag\\
        &=  C_{i+j}e^{-\ii E_{ij}/2} + C_{i-j}e^{\ii E_{ij}/2}\,.
    \end{align}
    In the second equality we have used the Weyl relations to extract the phase that depends on the causal propagator. Clearly, this reflects the non-Abelian nature of the Weyl algebra, and for the cases when the field operators commute, this reduces to the well-known trigonometric product-to-sum formula.
\end{proof}
With this, we can now compute the coefficients straightforwardly. For example, we have
\begin{align}
    \gamma_{cccc} &= \frac{1}{4}\omega(C_{\textsc{a}+\textsc{b}}^2+C_{\textsc{a}-\textsc{b}}^2+C_{\textsc{2a}}+C_{\textsc{2b}}\cos2E_{\textsc{ab}})\,,\\
    \gamma_{sccs} &=\frac{1}{4}\omega(S_{\textsc{a}+\textsc{b}}^2+S_{\textsc{a}-\textsc{b}}^2-C_{\textsc{2a}}+C_{\textsc{2b}}\cos2E_{\textsc{ab}})\,.
\end{align}
Using linearity of the algebraic state, the first coefficient in Eq.~\eqref{eq: channel-Bob-quasifree} reduces to
\begin{align}
    \gamma_{cccc}+\gamma_{sccs}&= \frac{1}{2}\rr{1+\nu_{\textsc{b}}\cos(2E_{\textsc{ab}})}\,,
\end{align}
where $\nu_\textsc{b} = \omega(W(2Ef_\textsc{b})) = e^{-2\mathsf{W}(f_{\textsc{b}},f_{\textsc{b}})}$. We can work out the rest:
\begin{align}
    \gamma_{cssc}+\gamma_{ssss} &= \frac{1}{2}\rr{1 - \nu_{\textsc{b}}\cos(2E_{\textsc{ab}})}\,,\\
    \gamma_{sscc}-\gamma_{cscs} &= \gamma_{scsc}-\gamma_{ccss} =  \frac{\ii}{2}\nu_\textsc{b}\sin(2E_\textsc{ab})\,,
\end{align}
hence we can write 
\begin{align}
    \Phi(\rbo) &= \frac{1 + \nu_{\textsc{b}}\cos(2E_{\textsc{ab}})}{2}\rbo + \frac{1 - \nu_{\textsc{b}}\cos(2E_{\textsc{ab}})}{2} \hat\mu_\textsc{b}\rbo\hat\mu_\textsc{b}\notag\\
    &\hspace{0.4cm}+\frac{\nu_\textsc{b}\sin(2E_\textsc{ab})}{2} [\rbo,\hat\mu_\textsc{b}]\,.
    \label{eq: channel-Bob-quasifree-2}
\end{align}

At this point, we are ready to see how the Fermi two-atom problem looks like. It is now clear that from Eq.~\eqref{eq: channel-Bob-quasifree-2}, the only way detector $A$ can influence detector $B$'s excitation is via the causal propagator $E_{\textsc{ab}}$. If $A$ and $B$ are spacelike separated, then $E_{\textsc{ab}}=0$, in which case we get
\begin{align}
    \Phi(\rbo) &\to  (1-p)\rbo + p\, \hat\mu_\textsc{b}\rbo\hat\mu_\textsc{b}\,.
    \label{eq: channel-Bob-quasifree-spacelike}
\end{align}
This is nothing but a Pauli channel with probability of applying Pauli operator $\hat\mu_\textsc{b}\equiv \hat\mu_\textsc{b}(\tau_{0,\textsc{b}})$
given by
\begin{align}
    p = \frac{1 - \nu_{\textsc{b}}}{2} = \frac{1 - e^{-2\mathsf{W}(f_\textsc{b},f_{\textsc{b}})}}{2}\,.
\end{align}
As an example, if $\hat\mu_\textsc{b}=\hat{\sigma}^x_\textsc{b}$ then this is just the bit-flip channel. What this means is that the channel only depends on the local field fluctuations around Bob's qubit, since the probability only depends on smeared Wightman two-point function  $\mathsf{W}(f_\textsc{b},f_\textsc{b})$ where $B$ is located. If we now pick the initial state of detector $B$ to be in the ground state, then the excitation probability\footnote{Actually, since the state $\Phi(\rbo)$ is in interaction picture, we should convert this to Schr\"odinger picture first. However, it does not matter since this conversion is a unitary constructed out of the free Hamiltonian of detector $B$ and it cannot change probabilities, so all the conclusions in what follows are unchanged.}
\begin{align}
    \Pr(\ket{g_\textsc{b}}\to\ket{e_\textsc{b}}) &= |\braket{e_\textsc{b}|\Phi(\ketbra{g_\textsc{b}}{g_\textsc{b}})|e_\textsc{b}}|^2
\end{align}
which is purely a function of $f_\textsc{b}$ and is independent of $f_\textsc{a}$. In fact, this conclusion is  \textit{independent} of the initial state of $B$: we only pick the initial state to match the Fermi two-atom problem. Furthermore, we are able to arrive at this conclusion without performing any integral or Dyson series expansion!

If the detectors are causally connected via the field, i.e., when $E_\textsc{ab}\neq 0$, then the story is different, since now the final state of detector depends on detector $A$'s support (recall that we assumed in the Fermi problem that the state of detector $A$ is fixed to be $\rao=\ketbra{e_\textsc{a}}{e_\textsc{a}}$). Hence indeed the excitation probability will now depend on both $f_\textsc{a}$ and $f_\textsc{b}$. This cleanly (re)affirms the causal behaviour of relativistic quantum field and its interactions with atomic-like systems, or any incarnation of light-matter interactions involving quantum fields.

We emphasize that we reached these conclusions about causal behaviour of the light-matter interaction with a quantum field \textit{without} performing a single integration over any field modes and fully non-perturbatively, hence the outcome in terms of the qubit channel \eqref{eq: channel-Bob-quasifree-2} is very transparent. Furthermore, notice that since we have expressed this in terms of causal propagator $E_\textsc{ab}\equiv E(f_\textsc{a},f_\textsc{b})$ without making any reference to a particular choice of globally hyperbolic spacetime $(\M,g_\textsc{ab})$, or any trajectory\footnote{We picked static trajectories because the time ordering is straightforward and we do not need general trajectories for Fermi problem, but really this choice needs not be made so long as one is careful about the spatial dependence of monopole operators.} or interaction profile $f_j$ of the detectors for that matter (although we picked instantaneous switching which is suited for this problem). Consequently, this conclusion remains true even if the spacetime has highly non-trivial causal structure, such as Schwarzschild background where the causal propagator cannot even be easily written in closed form \cite{Casals2020commBH}. We are not even choosing the kind of monopole operator of the detector: all we needed (and achieved) is to show that nothing about detector $A$ influences the density matrix of detector $B$ in as effortless manner as possible.

\section{Effect of projective measurement on Alice's detector}
\label{sec: measurement}

In this section we briefly discuss what happens when we allow detector-based measurement into the problem. Suppose that after interaction with the field, we want to measure detector $A$ and making sure that it is in the ground state. We perform projective measurement on detector $A$ which induces a POVM on the field; recall that we cannot perform projective measurement directly on the field as it violates relativistic causality \cite{sorkin1956impossible,josepolo2022measurement}. For simplicity we just consider the case when we project to the ground state $P=\ketbra{g_\textsc{a}}{g_\textsc{a}}$ while the detector is initially in excited state. 

Following the procedure outlined in \cite{josepolo2022measurement}, the density matrix of the field after projective measurement to detector $A$ becomes
\begin{align}
    {\tilde{\rho}}_\phi &= \frac{\tr_\textsc{ab}[(P\otimes\openone)U_\textsc{a}\roo U_\textsc{a}^\dagger(P\otimes\openone) ]}{\tr[(P\otimes\openone)U_\textsc{a}\roo U_\textsc{a}^\dagger]}\,.
\end{align}
For the choice of projectors, initial state, and the monopole $\hat\mu_\textsc{a}=\hat\sigma^x_\textsc{a}$, the expression simplifies greatly to
\begin{align}
    \tilde{\rho}_\phi &= \frac{S_\textsc{a}\rof S_\textsc{a}}{\tr \rof S_\textsc{a}^2}\,.
\end{align}
However, this updated state should only apply to observers in the causal future of detector $A$: if detector $B$ is spacelike separated from $A$, then Bob cannot learn about the measurement outcome of detector $B$ and the final state remains the same \cite{josepolo2022measurement}, i.e.,
\begin{align}
    \tilde{\rho}_\phi &= C_\textsc{a}\rof C_\textsc{a}+S_\textsc{a}\rof S_\textsc{a}\,,
\end{align}
which is essentially the statistical mixture of the two projective measurement outcomes. This state will give exactly the same result as if no measurement has been performed on Alice's detector.

In terms of algebraic state, the interpretation is somewhat simpler in the following sense: using the fact that $\omega$ and $\rho$ are related as $\omega(A)\equiv \tr(\hat\rho A)$ for Weyl element $A=W(Eg)$ with $\supp(g)\subseteq J^+(\supp(f_\textsc{a}))$, we can write 
\begin{align}
    \tilde{\omega}(A) &=
    \frac{\omega(S_\textsc{a}AS_\textsc{a})}{\omega(S_\textsc{a}^2)}\,,
\end{align}
otherwise $\tilde{\omega}(A)=\omega(A)$ when $\supp(g)\not\subseteq J^+(\supp(f_\textsc{a}))$. Here we assume that Alice performs the measurement directly after interaction with the field for convenience. What we are updating is simply the expectation values of the field observables in the causal future of the measurement process (as argued in \cite{josepolo2022measurement}). Note that this state update is not directly related to the causal propagator since in general $\supp(Ef_\textsc{a})\subset J^+(\supp(f_\textsc{a}))$ as a proper subset, i.e., we allow the possibility that Bob can learn about the measurement some other way (even classically without going through the quantum field). For example, Bob can be timelike-separated from Alice in Minkowski space (hence can learn about Alice's measurement in principle), but if the field used in the Fermi problem is massless then $E(f_\textsc{a},f_\textsc{b})=0$ also for timelike-separated regions, since the support of $E$ is localized along the null directions. The point is that the state update above is ``outside'' AQFT framework itself: for free field theory it is an extra rule induced from detector-based measurement theory.

If Bob is in the causal future of Alice, then by going through analogous computation to arrive at our earlier result and writing $\tilde\rho_{\textsc{ab}\phi}=\ketbra{e_\textsc{a}}{e_\textsc{a}}\otimes\rbo\otimes\tilde\rho_\phi$, 
we see that Bob's interaction with the field gives rise to a different channel
\begin{align}
    \tilde{\Phi}(\rbo) &= \tr_{\textsc{a}\phi}(U_\textsc{b}\tilde\rho_{\textsc{ab}\phi}U_\textsc{b}^\dagger)\notag\\
    &= \tilde{\omega}(C_\textsc{b}^2)\rbo + \tilde{\omega}(S_\textsc{b}^2)\hat \mu_\textsc{b}\rbo \hat\mu_\textsc{b}+ \ii\tilde{\omega}(S_\textsc{b}C_\textsc{b})\rbo \hat\mu_\textsc{b}  \notag\\
    &\hspace{1.75cm} -\ii \tilde{\omega}(C_\textsc{b}S_\textsc{b})\rbo  \hat\mu_\textsc{b}\,,
\end{align}
where $\supp(f_\textsc{b})
\subseteq J^+(\supp f_\textsc{a})$, but not necessarily $E(f_\textsc{a},f_\textsc{b})\neq 0$. In particular, supposing that $\hat\mu_\textsc{b}=\hat\sigma^x_\textsc{b}$, the probability of finding the detector in the excited state is given by
\begin{align}
    \tilde{\Pr}(\ket{g_\textsc{b}}\to\ket{e_\textsc{b}}) &= |\braket{e_\textsc{b}|\tilde{\Phi}(\ketbra{g_\textsc{b}}{g_\textsc{b}})|e_\textsc{b}}|^2 \notag\\
    &= |\tilde\omega(S_\textsc{b}^2)|^2 \notag\\
    &\neq \Pr(\ket{g_\textsc{b}}\to\ket{e_\textsc{b}})\,.
\end{align}
This difference arises because even if the causal propagator $E(f_\textsc{a},f_\textsc{b})=0$, the two channels with or without measurement are not the same, i.e., $\tilde{\Phi}(\rbo)\neq\Phi(\rbo)$. This has to do with the fact that Bob needs not learn about Alice's measurement outcome using the scalar field that the two atoms are involved in this setup. Having more information about Alice's state implies that Bob learns more about the state of the field before turning on detector $B$. So even if detector $B$ is not influenced by the excitation of detector $A$ that gets propagated by the field, Bob's result depends on whether he learns about Alice's measurement procedure. 

In any case, even with local measurement induced by the detector, the whole setting still respects Einstein causality simply because Bob cannot learn about Alice's measurement outcome without following relativistic principle. In particular, if Bob is spacelike separated from Alice's measurement process, then Bob can basically treat the situation as if measurement did not occur and the calculation in the previous section just carries forward as-is. We then simply recover the result in previous section where detector $B$ is completely independent of what happens to detector $A$.

\section{Discussion and outlook}
\label{sec: conclusion}

To summarize, we have reframed the Fermi two-atom problem using Unruh-DeWitt detector model, but this time we approached the problem using fully non-perturbative method via delta-coupling. Indeed, the delta-coupling approach is naturally suited for the Fermi problem. We showed that using the tools from (relativistic) quantum information theory, the calculation provides a much cleaner and transparent approach to demonstrate the causal nature of the detector-field interaction in Fermi two-atom problem that is (1) valid for arbitrary globally hyperbolic curved spacetimes, (2) independent of the detector configuration (monopole operator) of detector $B$, and (3) fully non-perturbatively, without any use of perturbation theory or complicated integration. All we used is simply the properties of the algebra of observables of the field theory, the choice of quasifree Hadamard state (which is the vacuum for the Fermi two-atom problem), and the causal relation between the two local  interaction regions $f_\textsc{a},f_\textsc{b}$ of each detector with the field.

It is worth stressing that in the physics community familiar with relativistic quantum field theory, this result --- that there is no causality violation in the Fermi two-atom problem --- is far from surprising and should be \textit{required} to be true. Furthermore, recall that the conclusion that the Fermi two-atom setup should be causal is not new: \cite{Cliche2010channel, Jonsson2014signalQED,Leon2011fermi,Causality2015Eduardo} are good representative examples where perturbative methods are computed reliably enough to arrive at the same conclusion. Indeed, in \cite{Cliche2010channel} the setup was shown to be causal to all orders in perturbation theory by treating the Fermi problem as a communication via quantum channel between the two parties. What we have done amounts to cleaning up the approach in a way that is nearly effortless, fully compatible with quantum information framework, and generalizes easily to curved spacetimes, where the causal propagator (smeared field commutator) can be very difficult to construct explicitly. 

Let us briefly go through several generalizations that could have been easily done in this work. First, note that we have only used the \textit{quasifree} nature of the field state. We did not even use directly the Hadamard property, since for the Fermi two-atom problem it is somewhat inconsequential what the value of the smeared Wightman two-point function $\mathsf{W}(f_\textsc{b},f_\textsc{b})$ is. Therefore, our results actually will still be true even if the field is not in the vacuum state: all that matters is that the probability of the Pauli channel $p$ only depends on $f_\textsc{b}$. The choice of vacuum state merely reflects the original setup of the Fermi two-atom problem. In particular, any other quasifree pure states with vanishing odd-point functions will automatically work verbatim without a single extra computation needed since it only changes the value of $p$: this includes, for instance, (squeezed) thermal states. It should be of no surprise, after this demonstration, that the Fermi two-atom problem should be generalizable to arbitrary physically reasonable (Hadamard) field states, although the computation may become more involved---the quasifree states (for which the vacuum state is included) made the calculation essentially effortless, with no integral to be performed to arrive at the result. 

The non-perturbative method adopted here also allowed for very nice inclusion of the recently proposed local measurement theory for QFT which is based on  positive-valued operator measure (POVM) induced by UDW detectors \cite{josepolo2022measurement}. Although not needed for the Fermi problem, we considered a slightly tweaked problem where we want to include measurement step: that is, what happens to $B$ if $A$ is \textit{measured} to be in the ground state after interaction with the field\footnote{Note that ``measurement'' here means something precise in quantum information in terms of measurement operators and POVM elements as understood in quantum information \cite{nielsen2000quantum,Wilde2013textbook}, not just computation of transition amplitudes as done in \cite{Fermi1932radiation,Hegerfeldt1994fermi}.}. The calculations done here gave us a cleaner calculation than \cite{josepolo2022measurement} when restricted to delta-coupling regime. While we did not consider more complicated situations with multiple measurements and more than two observers for Fermi problem, we believe that the approach adopted here would be useful for protocols where detector-based local measurement theory for the field would become important component. We leave this for future work.

Last but not least, it would also be interesting to explore the benefit of using such non-perturbative approach on non-standard settings, such as the scenarios considered in \cite{Menezes2017disordered} where the causal structure itself is somewhat fuzzy, or when quantum reference frames (QRFs) are employed.  We also see no obstruction to generalize this to arbitrary detector system such as qudits or harmonic oscillators, as well as any type of relativistic quantum fields: so long as the detectors interact locally with the field and the field obeys relativistic (Einstein) causality, it will lead to analogous conclusions. We leave these for future investigations as well. As far as we are concerned, we believe this is the most general setting we need to consider in the Fermi two-atom problem, and the use of quantum information language allowed us to provide a neater presentation of the problem for more quantum information-oriented readers.

\section*{Acknowledgment}

E. T. is grateful for insightful correspondence with Maximillian H. Ruep and in particular Tales Rick Perche for (very) useful discussions regarding Fermi normal coordinates. E. T. acknowledges that this work is made possible from the funding through both of his supervisors  Robert B. Mann and Eduardo Mart\'in-Mart\'inez. This work is supported in part by the Natural Sciences and Engineering Research Council of Canada (NSERC).  The funding support from Eduardo Mart\'in-Mart\'inez is through the Ontario Early Research Award and NSERC Discovery program. The work has been performed at the Institute for Quantum Computing, University of Waterloo, which is supported by Innovation, Science and Economic Development Canada. This work is partially conducted at University of Waterloo and Institute for Quantum Computing, which lie on the traditional territory of the Neutral, Anishinaabeg, and Haudenosaunee Peoples. The University of Waterloo and the Institute for Quantum Computing are situated on the Haldimand Tract, land that was promised to Six Nations, which includes six miles on each side of the Grand River. E. T. would also like to thank both supervisors for providing the opportunity to explore some research ideas independently, through which this work becomes possible.

\bibliography{Fermi}
\end{document}